\documentclass{IEEEtran}
\usepackage{cite}
\usepackage{graphicx}
\usepackage{textcomp}
\usepackage{subfigure}
\usepackage{epsfig} 
\usepackage{multicol}
\usepackage{amsmath,amsfonts,amsthm,bm}
\usepackage{amssymb}
\usepackage[ruled,norelsize,linesnumbered]{algorithm2e}
\usepackage{epstopdf}
\usepackage{tabularx}
\usepackage{comment}
\usepackage{caption}
\usepackage{soul,xcolor}
\usepackage{subcaption}
\usepackage{url}
\usepackage{float}
\usepackage{array}
\usepackage{multirow}

\newtheorem{theorem}{Theorem}
\newtheorem{remark}{Remark}
\newtheorem{lemma}{Lemma}
\newtheorem{assumption}{Assumption}
\newtheorem{definition}{Definition}

\def\argmin{\mathop{\rm argmin}}

\def\BibTeX{{\rm B\kern-.05em{\sc i\kern-.025em b}\kern-.08em
    T\kern-.1667em\lower.7ex\hbox{E}\kern-.125emX}}

\begin{document}
\title{Differentially-Private Distributed Model Predictive Control of Linear Discrete-Time Systems with Global Constraints}
\author{Kaixiang Zhang, \IEEEmembership{Member, IEEE}, Yongqiang Wang, \IEEEmembership{Senior Member, IEEE}, Ziyou Song, \IEEEmembership{Senior Member, IEEE}, Zhaojian Li, \IEEEmembership{Senior Member, IEEE}\vspace{-1.6em}
\thanks{This work was supported in part by National Science Foundation (NSF) under Grant 2045436. The work of Yongqiang Wang was supported in part by NSF under Grant 2219487. \textit{(Corresponding author: Zhaojian Li.)}}
\thanks{Kaixiang Zhang and Zhaojian Li are with the Department of Mechanical Engineering, Michigan State University, East Lansing, MI 48824, USA (e-mail: zhangk64@msu.edu, lizhaoj1@egr.msu.edu).}
\thanks{Yongqiang Wang is with the Department of Electrical and Computer Engineering, Clemson University, Clemson, SC 29634, USA (e-mail: yongqiw@clemson.edu).}
\thanks{Ziyou Song is with the Department of Electrical Engineering and Computer Science, University of Michigan, Ann Arbor, MI 48109, USA (e-mail: ziyou@umich.edu).}
}

\maketitle

\begin{abstract}
Distributed model predictive control (DMPC) has attracted extensive attention as it can explicitly handle system constraints and achieve optimal control in a decentralized manner. However, the deployment of DMPC strategies generally requires the sharing of sensitive data among subsystems, which may violate the privacy of participating systems. In this paper, we propose a differentially-private DMPC algorithm for linear discrete-time systems subject to coupled global constraints. Specifically, we first show that a conventional distributed dual gradient algorithm can be used to address the considered DMPC problem but cannot provide strong privacy preservation. Then, to protect privacy against the eavesdropper, we incorporate a differential-privacy noise injection mechanism into the DMPC framework and prove that the resulting distributed optimization algorithm can ensure both provable convergence to a global optimal solution and rigorous $\epsilon$-differential privacy. In addition, an implementation strategy of the DMPC is designed such that the recursive feasibility and stability of the closed-loop system are guaranteed. Simulation results are provided to demonstrate the effectiveness of the developed approach.
\end{abstract}


\vspace{-10pt}
\section{Introduction}
Over the past decades, model predictive control (MPC) has achieved great success due to its ability to explicitly handle system constraints and ensure desired control performance~\cite{Mayne2014AUTO}. MPC can be implemented in either a centralized or distributed manner. Centralized MPC requires a central unit to process all system information, making it computationally intensive and less scalable for large-scale systems. Consequently, distributed MPC (DMPC) has emerged as a promising alternative, offering the advantages of distributed systems and having been effectively applied in various areas~\cite{Hans2018TSE,Luis2020RAL}.

DMPC studies can be roughly categorized based on the type of couplings between subsystems: cost function couplings, system dynamics couplings, and constraint couplings. This paper focuses on systems with coupled global constraints, which have many real-world applications~\cite{BaiTVT2023,AlizadehSPM2012,Notarstefano2019FT}. Existing methods address coupled constraints through techniques like sequential optimization~\cite{Richards2007IJC} and parallel computation~\cite{Trodden2014SCL}. However, global optimality remains unclear in these approaches. 
To achieve global optimality, a DMPC scheme based on distributed alternating direction multiplier method (ADMM) is proposed in~\cite{Wang2017AUTO}. Other extensions include a push-sum dual gradient algorithm for time-varying directed networks~\cite{Jin2020TAC}, a noisy ADMM algorithm for handling communication noise~\cite{Li2021AUTO}, and a primal-dual algorithm designed with contraction theory~\cite{SuTCYB2022}.

The aforementioned methods~\cite{Wang2017AUTO,Jin2020TAC,Li2021AUTO,SuTCYB2022} employ distributed optimization to address DMPC problems, requiring subsystems to share local information to meet coupled global constraints. However, these shared messages and global constraints often contain sensitive data, raising concerns about privacy leakage. Eavesdroppers could wiretap communications and deduce private information, potentially leading to safety risks and economic losses. For example, in DMPC for automated vehicles~\cite{BaiTVT2023}, global constraints are formulated with each vehicle's position and velocity, which are sensitive and should be protected from disclosure. Similarly, in demand-side management for smart grids~\cite{AlizadehSPM2012}, global constraints capture the relationship between aggregated customer loads and the power bid, potentially revealing proprietary consumption patterns. As such, ensuring privacy protection in DMPC is essential for both security and practical deployment. While few results address privacy in DMPC, privacy-preserving methods for distributed optimization are well established. For the latter, a common technique is homomorphic encryption, which conceals sensitive information via cryptography~\cite{Lu2018AUTO,Zhang2018TCNS} and can be extended to DMPC~\cite{Zhao2023TCNS}.
However, this technique generally incurs high communication and computation overhead due to complex encryption and decryption processes. In contrast, differential privacy (DP) offers a lightweight alternative with strong theoretical guarantees. DP-based methods have been applied to distributed optimization by adding persistent noise to objective functions~\cite{Nozari2016TCNS} or shared information~\cite{Huang2015,Xiong2020TCNS,Ding2021TAC}. Nevertheless, the direct injection of DP noise to existing algorithms inevitably compromises optimization performance, resulting in a trade-off between accuracy and privacy. Note that extending DP-based methods to DMPC is particularly challenging, as the compromise on optimization accuracy can impair control performance and potentially violate constraints.

In this paper, a differentially-private DMPC algorithm is designed for linear discrete-time systems with coupled global constraints. We first demonstrate the need for privacy preservation by showcasing that a conventional distributed dual gradient algorithm for DMPC is vulnerable to eavesdropping attacks. A DP noise injection mechanism is then introduced into the distributed dual gradient algorithm to obscure exchanged private information. Leveraging results from~\cite{WangTAC2024,WangTAC2023}, we carefully design weakening factor and step-size sequences to effectively mitigate the influence of DP noise. Rigorous analysis shows that the proposed algorithm can ensure almost sure convergence to a global optimal solution and maintain $\epsilon$-differential privacy with a finite cumulative privacy budget. Aligned with the privacy-preserving distributed algorithm, we provide an implementation strategy for DMPC, ensuring the recursive feasibility and stability of the closed-loop system. Simulations are performed to validate the efficacy of the proposed scheme. Our work differs significantly from existing DMPC methods~\cite{Wang2017AUTO,Jin2020TAC,Li2021AUTO,SuTCYB2022} and DP-based privacy methods~\cite{Huang2015,Xiong2020TCNS,Ding2021TAC,WangTAC2024,WangTAC2023}. First, unlike~\cite{Wang2017AUTO,Jin2020TAC,Li2021AUTO,SuTCYB2022}, which primarily focus on control design and are susceptible to eavesdropping, we integrate a novel noise injection mechanism into the DMPC framework to address both control and privacy issues. Second, in contrast to algorithms in~\cite{Huang2015,Xiong2020TCNS,Ding2021TAC} that directly add DP noise to exchanged information and have to trade convergence accuracy for privacy, our algorithm uses weakening factor and step-size sequences to minimize the adverse impact of DP noise, achieving both a finite cumulative privacy budget and guaranteed convergence. Third, \cite{WangTAC2024,WangTAC2023} are designed for distributed optimization, and cannot be applied to DMPC problems. Our work borrows ideas from \cite{WangTAC2024,WangTAC2023} to construct tailored weakening factor and step-size sequences, and develops a DMPC-specific implementation strategy to extend the applicability of these techniques to the distributed control framework.

The rest of this paper is organized as follows. In Section~\ref{sec:problem}, the preliminaries of DMPC and DP are introduced. In Section~\ref{sec:DP}, a new differentially-private distributed dual gradient algorithm is developed, and convergence analysis is conducted. Section~\ref{sec:MPC} presents the implementation strategy of DMPC. Finally, a numerical study is given in Section~\ref{sec:simulation}, and concluding remarks are summarized in Section~\ref{sec:conclusion}.


{\bf Notations:}
$\mathbb{R}^n$ stands for the $n$-dimensional Euclidean space, and $\mathbb{R}_{+}^{n}$ is the non-negative orthant of $\mathbb{R}^n$. Given two integers $a$ and $b$ ($a< b$), $\mathbb{Z}_{a}^{b}$ represents the set $\{a, a+1, \cdots, b\}$. $I_n$ denotes the identity matrix of dimension $n$. ${\bf 1}_n$ and ${\bf 0}_n$ represent the $n$-dimensional column vector with all entries being 1 and 0, respectively. We use $Q>(\ge) 0$ to denote that $Q$ is a positive definite (semi-definite) matrix. $\|x\|$ and $\|x\|_{1}$ represent the standard Euclidean norm and the $L_{1}$ norm of a vector $x$, respectively. Moreover, $\|x\|_{Q}^{2}:=x^{\top}Qx$. The Euclidean projection of vector $x$ on a convex set $\mathcal{X}\subset \mathbb{R}^{n}$ is denoted by $\Pi_{\mathcal{X}} \left[ x \right]:=\argmin_{h\in \mathcal{X}} \|h-x\|$.

\vspace{-5pt}
\section{Problem Formulation and Preliminaries} \label{sec:problem}
\subsection{Problem Description}
Consider $M$ linear discrete-time subsystems, each described as follows:
\abovedisplayskip= 3pt 
\belowdisplayskip= 2.5pt
\begin{equation} \label{equ:LDT}
{\small
	\begin{aligned}
		x_{i}(t+1) = A_{i}x_{i}(t)+B_{i}u_{i}(t), \quad i\in \mathbb{Z}_{1}^{M},
	\end{aligned}
}
\end{equation}
where $x_{i}(t)\in \mathbb{R}^{n_{i}}$ and $u_{i}(t)\in \mathbb{R}^{m_{i}}$ are the state and control input of subsystem $i$ at time instant $t$, respectively. Each subsystem $i$ should satisfy local constraints $x_{i}(t)\in \mathcal{X}_{i}$ and $u_{i}(t)\in \mathcal{U}_{i}$, with $\mathcal{X}_{i} \subset \mathbb{R}^{n_{i}}$ and $\mathcal{U}_{i}\subset \mathbb{R}^{m_{i}}$ being state and input constraint sets, respectively. Moreover, all the subsystems are subject to $p$ global constraints described by
\begin{equation} \label{equ:globalConstraint}
{\small
	\begin{aligned}
		\sum_{i=1}^{M} (\varPsi_{x_{i}}x_{i}(t)+\varPsi_{u_{i}}u_{i}(t))\le {\bf 1}_{p},
	\end{aligned}
}
\end{equation}
where $\varPsi_{x_{i}}\!\in \! \mathbb{R}^{p\times n_i}$ and $\varPsi_{u_{i}}\!\in \! \mathbb{R}^{p\times m_i}$ are some given~matrices. The coupled linear constraints arise in many practical multi-agent systems, such as safety limits in automated vehicles~\cite{BaiTVT2023} and demand constraints in smart grids~\cite{AlizadehSPM2012}.
\begin{assumption} \label{assump:system}
Each subsystem, i.e., $(A_{i}, B_{i})$, is controllable. Additionally, $\mathcal{X}_{i}$ and $\mathcal{U}_{i}$ are bounded and closed polytopes which contain the origins as their inner point.
\end{assumption}

The DMPC problem is formulated as~\cite{Wang2017AUTO,Jin2020TAC,Li2021AUTO,SuTCYB2022}
\begin{subequations} \label{equ:DMPC}
\begin{align}
	\mathcal{P}:\quad &\min_{\{\tilde{{\bm{u}}}_{1}, \cdots, \tilde{{\bm{u}}}_{M}\}} \sum_{i=1}^{M}J_{i}(x_{i}(t), \tilde{{\bm{u}}}_{i}) \label{equ:DMPC_J}
	\\ 
	&\; \mathrm{s.t.}\quad \tilde{{\bm{u}}}_{i}\in \tilde{\mathcal{U}_{i}}(x_{i}(t)), \quad \sum_{i=1}^{M} f_{i}(x_{i}(t), \tilde{{\bm{u}}}_{i}) \le b(\varepsilon). 
	\label{equ:DMPC_local}
\end{align}
\end{subequations}
In \eqref{equ:DMPC_J}, $J_{i}(x_{i}(t), \tilde{{\bm{u}}}_{i})$ is the local objective function, which is defined as
\begin{equation}
{\small
	\begin{aligned}
		J_{i}(x_{i}(t), \tilde{{\bm{u}}}_{i})\!:=\!\sum_{\ell=0}^{N-1} \! \left( \|\tilde{x}_{i}(\ell|t)\|_{Q_{i}}^{2} \!+\! \|\tilde{u}_{i}(\ell|t)\|_{R_{i}}^{2} \right) \!+\! \|\tilde{x}_{i}(N|t)\|_{P_{i}}^{2},
	\end{aligned}
}
\end{equation}
where $N\in \mathbb{Z}_{>0}$ is the length of prediction horizon, $\tilde{x}_{i}(\ell|t)$ and $\tilde{u}_{i}(\ell|t)$ are the $\ell$th step predicted state and control input at time instant $t$, respectively, $\tilde{{\bm{u}}}_{i}:=\{\tilde{u}_{i}(0|t), \cdots, \tilde{u}_{i}(N-1|t)\}$ stands for the predicted input sequence over the prediction horizon, and $Q_{i}>0$, $R_{i}>0$, and $P_{i}>0$ are weight matrices. For each subsystem $i$, $P_{i}$ is the solution of the following algebraic Riccati equation: 
\begin{equation} \label{equ:ARE}
{\small
	\begin{aligned}
		(A_{i}+B_{i}K_{i})^{\top}P_{i}(A_{i}+B_{i}K_{i})-P_{i} = -(Q_{i}+K_{i}^{\top}R_{i}K_{i}),
	\end{aligned}
}
\end{equation}
where $K_{i}:=-(R_i+B_{i}^{\top}P_{i}B_{i})^{-1}B_{i}^{\top}P_{i}A_{i}$.
The local constraint set $\tilde{\mathcal{U}_{i}}(x_{i}(t))$ in~\eqref{equ:DMPC_local} is formulated as
\begin{equation}\label{equ:U}
{\small
	\begin{aligned}
		&\tilde{\mathcal{U}_{i}}(x_{i}(t)):= \{\tilde{{\bm{u}}}_{i}\in\mathbb{R}^{m_{i}N}:
		\tilde{x}_{i}(\ell+1|t) = A_{i}\tilde{x}_{i}(\ell|t)+B_{i}\tilde{u}_{i}(\ell|t), 
		\\
		&\tilde{x}_{i}(0|t)\!=\!x_{i}(t), \tilde{x}_{i}(\ell|t)\!\in \! \mathcal{X}_{i}, \tilde{u}_{i}(\ell|t)\!\in \! \mathcal{U}_{i}, \tilde{x}_{i}(N|t)\!\in \! \mathcal{X}_{i}^{f}, \ell \!\in \! \mathbb{Z}_{0}^{N-1}\},
	\end{aligned}
}
\end{equation}
with $\mathcal{X}_{i}^{f}$ being the terminal constraint set. In addition, the global coupled constraint in~\eqref{equ:DMPC_local} is a tightened form of the constraint in~\eqref{equ:globalConstraint}, and $f_{i}(x_{i}(t), \tilde{{\bm{u}}}_{i})$ and $b(\varepsilon)$ are given by
\begin{equation} \label{equ:fb}
{\small
	\begin{aligned}
		f_{i}(x_{i}(t), \tilde{{\bm{u}}}_{i}) &:= \begin{bmatrix}
			\varPsi_{x_{i}}\tilde{x}_{i}(0|t) +\varPsi_{u_{i}}\tilde{u}_{i}(0|t) \\
			\vdots \\
			\varPsi_{x_{i}}\tilde{x}_{i}(N-1|t) +\varPsi_{u_{i}}\tilde{u}_{i}(N-1|t)
		\end{bmatrix}, 
		\\
		b(\varepsilon) &:= 
		\begin{bmatrix}
			(1-\varepsilon M){\bf 1}_{p}^{\top}, \cdots,
			(1-\varepsilon MN){\bf 1}_{p}^{\top}
		\end{bmatrix}^{\top},
	\end{aligned}
}
\end{equation}
where $0\le \varepsilon< \frac{1}{MN}$ is a tolerance parameter. To facilitate the feasibility and stability analysis of DMPC, the terminal constraint set $\mathcal{X}_{i}^{f}$ is selected to satisfy
\begin{equation} \label{equ:terminal}
{\small
\begin{aligned}
	& K_{i}x_{i}\in \mathcal{U}_{i}, \quad (A_{i}+B_{i}K_{i})x_{i}\in \mathcal{X}_{i}^{f}, 
	\\
	& \sum_{i=1}^{M} (\varPsi_{x_{i}}\!+\!\varPsi_{u_{i}}K_{i}) x_{i} \!\le\! (1-\varepsilon MN) {\bf 1}_{p}, \, \forall x_{i}\in \mathcal{X}_{i}^{f}, \forall i \in \mathbb{Z}_{1}^{M}.
\end{aligned}
}
\end{equation}
For further details on the constraint tightening in~\eqref{equ:fb} and the construction of the terminal constraint set $\mathcal{X}_{i}^{f}$ to satisfy~\eqref{equ:terminal}, please refer to~\cite{Wang2017AUTO}.

\vspace{-5pt}
\begin{assumption} \label{assump:slater}
For the initial state $\{ x_{1}(0), \cdots, x_{M}(0) \}$, the Slater condition holds, i.e., there exists $\{\tilde{{\bm{u}}}_{1}, \cdots, \tilde{{\bm{u}}}_{M}\}$ that satisfies \eqref{equ:DMPC_local}.
\end{assumption}
\vspace{-5pt}
The communication network of $M$ subsystems is described by an interaction weight matrix $L=\{L_{ij}\}\in \mathbb{R}^{M\times M}$. Specifically, for each subsystem $i$, the neighbor set $\mathcal{N}_i$ consists of all subsystems $j$ that can directly communicate with subsystem $i$. If $j\in \mathcal{N}_{i}$, then $L_{ij}>0$; otherwise, $L_{ij}=0$. We define $L_{ii}:=-\sum_{j\in\mathcal{N}_i}L_{ij}$  for all $i\in \mathbb{Z}_{1}^{M}$. 
\vspace{-5pt}
\begin{assumption}\label{assump:L}
The matrix  $L$ is symmetric and satisfies
${\bf 1}_{M}^{\top}L={\bf
	0}_{M}^{\top}$, $L{\bf 1}_{M}={\bf
	0}_{M}$, and $ \|I_{M}+L-\frac{{\bf 1}_{M}{\bf 1}_{M}^{\top}}{M}\|<1$.
\end{assumption}
\vspace{-4pt}
Assumption~\ref{assump:L} guarantees that the communication network described by $L$ is connected, meaning that there exists a path from any subsystem to any other subsystem.

\vspace{-5pt}
\subsection{Distributed Dual-Gradient Method}
The Lagrangian function corresponding to the optimization problem in~\eqref{equ:DMPC} is given by $\mathcal{L}(\{\tilde{{\bm{u}}}_{i}\}, \lambda) = \sum_{i=1}^{M} J_{i}(x_{i}(t), \tilde{{\bm{u}}}_{i}) + \lambda^{\top}\left( \sum_{i=1}^{M} f_{i}(x_{i}(t), \tilde{{\bm{u}}}_{i}) - b(\varepsilon) \right)=\sum_{i=1}^{M} \left( J_{i}(x_{i}(t), \tilde{{\bm{u}}}_{i}) + \lambda^{\top}g_{i}(\tilde{{\bm{u}}}_{i}) \right)$, where $\lambda\in \mathbb{R}^{Np}_{+}$ is the Lagrangian multiplier and $g_{i}(\tilde{{\bm{u}}}_{i}):= f_{i}(x_{i}(t), \tilde{{\bm{u}}}_{i}) - \frac{b(\varepsilon)}{M}$. The dual problem of~\eqref{equ:DMPC} is defined as
\begin{equation} \label{equ:dual}
{\small
	\begin{aligned}
		\max_{\lambda\ge 0} \min_{\{ \tilde{{\bm{u}}}_{i}\in \tilde{\mathcal{U}_{i}}(x_{i}(t)) \}} \mathcal{L}(\{\tilde{{\bm{u}}}_{i}\}, \lambda).
	\end{aligned}
}
\end{equation}

Under Assumptions~\ref{assump:system}, \ref{assump:slater}, strong duality holds for~\eqref{equ:DMPC}, and the optimization problem~\eqref{equ:DMPC} can be solved via its dual formulation~\eqref{equ:dual}.   
In addition, the Saddle-Point Theorem holds, i.e., given an optimal primal-dual pair $(\{\tilde{{\bm{u}}}_{i}^{*}\}, \lambda^{*})$, the following relationship holds for any $\lambda\in \mathbb{R}^{Np}_{+}$ and $\tilde{{\bm{u}}}_{i}\in \tilde{\mathcal{U}_{i}}(x_{i}(t))$:
\begin{equation} \label{equ:saddle}
{\small
	\begin{aligned}
		\mathcal{L}(\{\tilde{{\bm{u}}}_{i}^{*}\}, \lambda) \le \mathcal{L}(\{\tilde{{\bm{u}}}_{i}^{*}\}, \lambda^{*}) \le \mathcal{L}(\{\tilde{{\bm{u}}}_{i}\}, \lambda^{*}).
	\end{aligned}
}
\end{equation}

A standard approach for solving~\eqref{equ:dual} is the distributed dual-gradient method~\cite{Falsone2017AUTO,Notarstefano2019FT}. Specifically, the dual variable $\lambda$ is treated as a consensus variable, and each subsystem has a local copy $\lambda_{i}^{k}$. $\Pi_{\mathbb{R}^{Np}_{+}} \left[ \cdot \right]$ denotes Euclidean projection of a vector on the set $\mathbb{R}^{Np}_{+}$, and $\gamma^{k}>0$ is the step-size. Then, the distributed dual-gradient method is summarized in Algorithm~\ref{algo:dual}, and the overall DMPC implementation is detailed in  Algorithm~\ref{algo:DMPC}.

In Algorithm~\ref{algo:dual}, each subsystem avoids sharing the primal variable and only shares its local copy $\lambda_{i}^{k}$ of the dual variable with its neighbors. However, this sharing mechanism cannot provide strong privacy protection, as the iteration trajectory of $\lambda_{i}^{k}$ still bears information of the primal variable. In particular, since the communication network $L$ and the step-size $\gamma^{k}$ are public information (otherwise the algorithm cannot be implemented in a fully decentralized manner), if an adversary can intercept all information exchanged in communication channels, it can record the updates of $\tilde{\lambda}_{i}^{k}$ and $\lambda_{i}^{k}$ at each iteration. Using consecutive updates $\tilde{\lambda}_{i}^{k}$ and $\lambda_{i}^{k+1}$, along with $\gamma^{k}$, the adversary can use~\eqref{equ:DDGA_3} to estimate $g_{i}(\tilde{{\bm{u}}}_{i}^{k+1})$. The value of $g_{i}(\tilde{{\bm{u}}}_{i}^{k+1})$ is privacy-sensitive, as it depends on the primal variable and is used to formulate the coupled global constraint. Therefore, it is necessary to incorporate a privacy protection mechanism into the distributed dual-gradient algorithm to ensure rigorous privacy protection in DMPC.
\setlength{\textfloatsep}{0pt}
\begin{algorithm}[!t]
	\small
	\SetAlFnt{\small}
	\SetKwInOut{Parameter}{Parameter}
	\SetKwInOut{Output}{Output}
	\caption{Distributed Dual-gradient Algorithm}
	\label{algo:dual}
	\SetAlgoLined
	\KwIn{$x_{i}(t)$, $i\in \mathbb{Z}_{1}^{M}$}
	\KwOut{$\tilde{{\bm{u}}}_{i}^{\bar{k}}$, $i\in \mathbb{Z}_{1}^{M}$}
	\textit{Initialization:} set $\lambda_{i}^{0}\in \mathbb{R}^{Np}_{+}$ and $\tilde{{\bm{u}}}_{i}^{0}\in \tilde{\mathcal{U}_{i}}(x_{i}(t))$, $\forall i\in \mathbb{Z}_{1}^{M}$
	\textit{Parameters:} deterministic sequence $\gamma^{k}>0$
	
	\For{$k=0, 1, \cdots, \bar{k}-1$}{
		\For{all $i\in \mathbb{Z}_{1}^{M}$ (in parallel)}{
			Every subsystem $i$ sends $\lambda_{i}^{k}$ to subsystem $j\in \mathcal{N}_{i}$;\\
			After receiving $\lambda_{j}^{k}$ from all $j\in \mathcal{N}_{i}$, subsystem $i$ updates its primal and dual variables:
			\begin{align}
				&\tilde{\lambda}_{i}^{k} = \lambda_{i}^{k}+\sum_{j\in \mathcal{N}_{i}} L_{ij}(\lambda_{j}^{k}-\lambda_{i}^{k}); \label{equ:DDGA_1}
				\\
				&\tilde{{\bm{u}}}_{i}^{k\!+\!1} \!=\! \argmin_{\tilde{{\bm{u}}}_{i}\in \tilde{\mathcal{U}_{i}}(x_{i}(t))}\!\! J_{i}(x_{i}(t), \tilde{{\bm{u}}}_{i}) \!+\! (\tilde{\lambda}_{i}^{k})^{\top}\!g_{i}(\tilde{{\bm{u}}}_{i}); \label{equ:DDGA_2}
				\\
				&\lambda_{i}^{k+1} \!= \Pi_{\mathbb{R}^{Np}_{+}} \left[ \tilde{\lambda}_{i}^{k} + \gamma^{k}g_{i}(\tilde{{\bm{u}}}_{i}^{k+1}) \right]; \label{equ:DDGA_3}
			\end{align}
		} 
	}			
\end{algorithm}

\setlength{\textfloatsep}{0pt}
\setlength{\floatsep}{0pt}
\begin{algorithm}
	\small
	\SetAlFnt{\small}
	\SetKwInOut{Parameter}{Parameter}
	\SetKwInOut{Output}{Output}
	\caption{DMPC Algorithm}
	\label{algo:DMPC}
	\SetAlgoLined
	At time instant $t$, every subsystem $i$ measures its state $x_{i}(t)$;
	
	Every subsystem $i$ computes $\tilde{{\bm{u}}}_{i}^{\bar{k}}$ by following Algorithm~\ref{algo:dual};
	
	Set the input sequence as $\tilde{{\bm{u}}}_{i}(t)=\tilde{{\bm{u}}}_{i}^{\bar{k}}$;
	
	Apply $\tilde{u}_{i}(0|t)$ to subsystem $i$; 
	
	Wait for the next time instant; let $t=t+1$ and go to step 1.			
\end{algorithm}

\vspace{-5pt}
\subsection{On Differential Privacy}
In this work, DP is used to characterize and quantify the achieved privacy level of distributed optimization algorithms. Drawing inspiration from the distributed optimization framework proposed by \cite{Huang2015}, we represent the DMPC problem in~\eqref{equ:DMPC} by four parameters $(L, \mathcal{J}, \tilde{\mathcal{U}}, \mathcal{G})$ to facilitate DP analysis. Specifically, $L$ is the interaction weight matrix describing the communication network, $\mathcal{J}:=\{ J_1,\, \cdots, J_M\}$ denotes the set of objective functions for individual subsystems, $\tilde{\mathcal{U}}:=\{\tilde{\mathcal{U}_{1}}, \cdots, \tilde{\mathcal{U}_{M}} \}$ is the domain of optimization variables, and $\mathcal{G}:=\{g_1, \cdots, g_{M}\}$ represents the set of constraint functions for individual subsystems. The adjacency between two optimization problems is defined as follows: 

\vspace{-5pt}
\begin{definition}\label{de:adjacency}
Two distributed optimization problems $\mathcal{P}=(L, \mathcal{J}, \tilde{\mathcal{U}}, \mathcal{G})$ and $\mathcal{P}^{'}=(L^{'}, \mathcal{J}^{'}, \tilde{\mathcal{U}}^{'}, \mathcal{G}^{'})$ are adjacent if they satisfy the following conditions: $\left. 1 \right)$ $L=L^{'}$, $\mathcal{J}=\mathcal{J}^{'}$, and $\tilde{\mathcal{U}}=\tilde{\mathcal{U}}^{'}$; $\left. 2 \right)$ There exists an $i\in \mathbb{Z}_{1}^{M}$ such that $ g_{i}\neq g_{i}^{'}$, and $g_{j}= g_{j}^{'}$ for all $j\in \mathbb{Z}_{1}^{M},\,j\neq i$; $\left. 3 \right)$ $g_{i}$ and $g_{i}^{'}$, while different, exhibit similar behaviors near $\theta^{*}$, where $\theta^{*}$ is the solution of $\mathcal{P}$.
More precisely, there exists a $\delta>0$ such that for all $\bm{u}_{i}$ and $\bm{u}_{i}^{'}$ within the domain $B_{\delta}(\theta^{*}):=\{\bm{v}: \bm{v}\in \mathbb{R}^{Nm_{i}}, \| \bm{v}-\theta^{*} \|<\delta \}$, $g_{i}(\bm{u}_{i})=g_{i}^{'}(\bm{u}_{i}^{'})$ holds.
\end{definition}

We denote the execution of a distributed optimization algorithm as $\mathcal{A}$, represented by a sequence of the iteration variable $\vartheta$, i.e., $\mathcal{A}=\{\vartheta^0,\vartheta^1,\cdots\}$. Assuming adversaries have access to all communicated messages among subsystems, their observation under an execution $\mathcal{A}$ is the sequence of these messages, denoted by $\mathcal{O}$.
Let $\mathbb{O}$ represent the set of all possible observation sequences. For a distributed optimization problem $\mathcal{P}$ with an initial state $\vartheta^0$, the observation mapping is defined as $\mathcal{R}_{\mathcal{P},\vartheta^0}(\mathcal{A}):= \mathcal{O}$. Moreover, given $\mathcal{P}$, $\vartheta^0$, and an observation sequence $\mathcal{O}$,  $\mathcal{R}_{\mathcal{P},\vartheta^0}^{-1}(\mathcal{O})$ denotes the set of executions $\mathcal{A}$ that could generate the observation $\mathcal{O}$.
\vspace{-5pt}
\begin{definition}[$\epsilon$-differential privacy, \cite{Huang2015}] \label{de:differential_privacy}
For a given $\epsilon>0$, an iterative distributed algorithm ensures $\epsilon$-differential privacy if for any two adjacent optimization problems $\mathcal{P}$ and $\mathcal{P}'$, any initial state ${\vartheta}^0$, and any set of observation sequences $\mathcal{O}_s\subseteq\mathbb{O}$, the following relationship always holds:
\begin{equation}
	{\small
		\begin{aligned}
			\mathbb{P}[\mathcal{R}_{\mathcal{P},\vartheta^0}\left(\mathcal{O}_s\right)]\leq e^\epsilon\mathbb{P}[\mathcal{R}_{\mathcal{P}^{'},\vartheta^0}\left( \mathcal{O}_s\right)],
		\end{aligned}
	}
\end{equation}
with the probability $\mathbb{P}$ taken over the randomness of iteration processes.
\end{definition}

The definition of $\epsilon$-DP guarantees that adversaries, with access to all communicated information, cannot infer knowledge about any participating subsystem's sensitive information.

\section{Differentially-Private Distributed Dual-Gradient Algorithm} \label{sec:DP}
\vspace{-4pt}
\subsection{Algorithm Description}
In this section, a DP noise injection mechanism is proposed to achieve privacy preservation in the distributed dual-gradient algorithm. The method is summarized in Algorithm~\ref{algo:dual_privacy}. 

In contrast to Algorithm~\ref{algo:dual}, where each subsystem directly sends $\lambda_{i}^{k}$ to its neighbors, Algorithm~\ref{algo:dual_privacy} adds DP noise $\zeta_{i}^{k}$ to $\lambda_{i}^{k}$ and shares the perturbed signal $\hat{\lambda}_{i}^{k}\!:=\!\lambda_{i}^{k}\!+\!\zeta_{i}^{k}$ among the communication network. Thus, the adversary's available information is the sequence $\{\hat{\lambda}_{i}^{k}\}$. The randomness introduced by the DP noise ensures that extracting meaningful information from $\{\hat{\lambda}_{i}^{k}\}$ is statistically impossible. Moreover, it should be noted that directly integrating persistent DP noise into optimization algorithms will compromise the convergence accuracy. To address this issue, we utilize findings from~\cite{WangTAC2024,WangTAC2023} to design diminishing weakening factor sequence $\{\chi^k\}$ and step-size sequence $\{\gamma^k\}$. As shown in~\eqref{equ:DP_DDGA_2}, $\chi^{k}$ and $\gamma^k$ are applied on the terms ($L_{ij}(\hat{\lambda}_{j}^{k}\!-\!\lambda_{i}^{k})$) and $g_{i}(\tilde{{\bm{u}}}_{i}^{k+1})$, respectively. The principle behind incorporating the diminishing weakening factor and step-size sequences is to gradually eliminate the impact of DP noise, thereby ensuring convergence accuracy.

\vspace{-5pt}
\begin{algorithm}[!h]
	\small
	\SetAlFnt{\small}
	\SetKwInOut{Parameter}{Parameter}
	\SetKwInOut{Output}{Output}
	\caption{Differentially-private Distributed Dual-gradient Algorithm}
	\label{algo:dual_privacy}
	\SetAlgoLined
	\KwIn{$x_{i}(t)$, $i\in \mathbb{Z}_{1}^{M}$}
	\KwOut{$\tilde{{\bm{u}}}_{i}^{\bar{k}}$, $i\in \mathbb{Z}_{1}^{M}$}
	\textit{Initialization:} set $\lambda_{i}^{0}\in \mathbb{R}^{Np}_{+}$ and $\tilde{{\bm{u}}}_{i}^{0}\in \tilde{\mathcal{U}_{i}}(x_{i}(t))$, $\forall i\in \mathbb{Z}_{1}^{M}$
	\textit{Parameters:} deterministic sequence $\gamma^{k}>0$ and $\chi^{k}>0$
	
	\For{$k=0, 1, \cdots, \bar{k}-1$}{
		\For{all $i\in \mathbb{Z}_{1}^{M}$ (in parallel)}{
			Every subsystem $i$ adds DP noise $\zeta_{i}^{k}$ to $\lambda_{i}^{k}$, and then sends the obscured value $\hat{\lambda}_{i}^{k}:=\lambda_{i}^{k}+\zeta_{i}^{k}$ to subsystem $j\in \mathcal{N}_{i}$;\\
			After receiving $\hat{\lambda}_{j}^{k}$ from all $j\in \mathcal{N}_{i}$, subsystem $i$ updates its primal and dual variables:
			\begin{align}
				&\tilde{{\bm{u}}}_{i}^{k+1} = \argmin_{\tilde{{\bm{u}}}_{i}\in \tilde{\mathcal{U}_{i}}(x_{i}(t))} J_{i}(x_{i}(t), \tilde{{\bm{u}}}_{i}) + (\lambda_{i}^{k})^{\top}g_{i}(\tilde{{\bm{u}}}_{i}); \label{equ:DP_DDGA_1}
				\\
				&\begin{aligned}
					\lambda_{i}^{k+1} = \Pi_{\mathbb{R}^{Np}_{+}} [ \lambda_{i}^{k}&+\chi^{k} \sum_{j\in \mathcal{N}_{i}} L_{ij}(\hat{\lambda}_{j}^{k}-\lambda_{i}^{k})
					\\
					&+ \gamma^{k}g_{i}(\tilde{{\bm{u}}}_{i}^{k+1}) ]; \label{equ:DP_DDGA_2}
				\end{aligned}
			\end{align}
		} 
	}			
\end{algorithm}
\vspace{-5pt}
To facilitate the convergence and privacy analysis, the following DP noise assumption is introduced:
\vspace{-5pt}
\begin{assumption}\label{assump:dp-noise}
For every $k$ and every $i\in \mathbb{Z}_{1}^{M}$, conditional on $\lambda_{i}^{k}$,
the DP noise $\zeta_i^k$ satisfies $
\mathbb{E}\left[\zeta_i^k\mid \lambda_i^k\right]=0$ and $\mathbb{E}\left[\|\zeta_i^k\|^2\mid \lambda_{i}^{k}\right]=(\sigma_{i}^k)^2$ for all  $k\ge0$, and
\begin{equation}\label{eq:condition_assumption1}
	{\small
		\begin{aligned}
			\sum_{k=0}^\infty (\chi^k)^2\, \max_{i\in \mathbb{Z}_{1}^{M}}(\sigma_{i}^k)^2 <\infty,
		\end{aligned}
	}
\end{equation} 
where $\{\chi^k\}$ ($\chi^k > 0$) is the weakening factor sequence from Algorithm~\ref{algo:dual_privacy}.
\end{assumption}
\vspace{-5pt}
Considering Assumption~\ref{assump:dp-noise}, we use the Laplace noise mechanism to generate $\zeta_{i}^{k}$ and then add it to all shared messages. More specifically, given a constant $\nu>0$, let ${\rm Lap}(\nu)$ represent a Laplace distribution of a scalar random variable, and $\rho \rightarrow \frac{1}{2\nu}e^{-\frac{|\rho|}{\nu}}$ be the corresponding probability density function. At each iteration $k$, every element of $\zeta_{i}^{k}$ is independently sampled from Laplace distribution ${\rm Lap}(\nu^{k})$, where $\nu^{k}>0$. One can verify that the mean and variance of ${\rm Lap}(\nu^{k})$ is zero and $2(\nu^{k})^{2}$, respectively. Therefore, $\zeta_{i}^{k}$ satisfies $
\mathbb{E}\left[\zeta_i^k\mid \lambda_i^k\right]=0$ and $\mathbb{E}\left[\|\zeta_i^k\|^2\mid \lambda_{i}^{k}\right]=(\sigma_{i}^k)^2=2(\nu^{k})^{2}$.
\begin{remark} \label{remark:DP-noise}
In Algorithm~\ref{algo:dual_privacy}, the variance of DP noise $\zeta_{i}^{k}$, i.e., $2(\nu^{k})^{2}$, can be constant or increasing with $k$. To satisfy condition~\eqref{eq:condition_assumption1}, one can carefully design the weakening factor sequence $\{\chi^{k}\}$ to make its decreasing rate outweigh the increasing rate of the noise level sequence $\{\nu^{k}\}$. For instant, \eqref{eq:condition_assumption1} is satisfied with $\chi^{k} = \frac{c_1}{1+c_{2}k^{c_3}}$ and $\nu^{k}=d_{1}+d_{2}k^{d_{3}}$, where $c_{1}>0$, $c_{2}>0$, $0.5<c_{3}<1$, $d_{1}>0$, $d_{2}>0$, and $0<d_{3}\le1-c_{3}$.
For notation simplicity, we assume that all subsystems use the same Laplace distribution ${\rm Lap}(\nu^{k})$ to generate DP noise. Actually, each subsystem can independently select its DP noise intensity as long as condition~\eqref{eq:condition_assumption1} is met.
\end{remark}
\vspace{-15pt}
\subsection{Convergence Analysis}
The arithmetic average of local dual variables $\lambda_{i}^{k}$ is given~by
\begin{equation}
{\small
	\begin{aligned}
		\bar{\lambda}^{k} = \frac{1}{M}\sum_{i=1}^{M}\lambda_{i}^{k}.
	\end{aligned}
}
\end{equation}
The relation between $\lambda_{i}^{k}$ and $\bar{\lambda}^{k}$ is summarized in the following theorem.
\vspace{-5pt}
\begin{theorem}\label{Th:consensus_tracking}
Suppose Assumptions~\ref{assump:system}, \ref{assump:L}, and \ref{assump:dp-noise} hold. If
the non-negative weakening factor sequence $\{\chi^k\}$ and the step-size sequence $\{\gamma^k\}$ in Algorithm~\ref{algo:dual_privacy} satisfy $\sum_{k=0}^{\infty}\chi^k=\infty$, $\sum_{k=0}^{\infty}(\chi^k)^2<\infty$, and $\sum_{k=0}^{\infty}\frac{(\gamma^k)^2}{\chi^k}<\infty$, then the following results hold almost surely: $\left. 1 \right)$ $\lim_{k\to\infty}\|\lambda_i^k - \bar{\lambda}^{k}\|=0$ for all $i\in \mathbb{Z}_{1}^{M}$; $\left.2 \right)$ $\sum_{k=0}^{\infty}\chi^k\sum_{i=1}^{M}\|\lambda_i^k-\bar{\lambda}^k\|^2<\infty$; $\left.3 \right)$ $\sum_{k=0}^{\infty}\gamma^k \sum_{i=1}^{M}\|\lambda_i^k-\bar{\lambda}^{k}\| <\infty$.
\end{theorem}

\begin{proof}
As stated in Assumption~\ref{assump:system}, $\mathcal{X}_{i}$ and $\mathcal{U}_{i}$ are bounded, and then it can be concluded from~\eqref{equ:U} that the local constraint set $\tilde{\mathcal{U}_{i}}(x_{i}(t))$ is bounded. From~\eqref{equ:fb} and the relation $g_{i}(\tilde{{\bm{u}}}_{i}):= f_{i}(x_{i}(t), \tilde{{\bm{u}}}_{i}) - \frac{b(\varepsilon)}{M}$, we have that for any $\tilde{{\bm{u}}}_{i}\in \tilde{\mathcal{U}_{i}}(x_{i}(t))$, $g_{i}(\tilde{{\bm{u}}}_{i})$ is bounded, i.e., there exists a constant $C_{g}\in \mathbb{R}_{+}$ such that $\|g_{i}(\tilde{{\bm{u}}}_{i}) \|\le C_{g}, \forall \tilde{{\bm{u}}}_{i}\in \tilde{\mathcal{U}_{i}}(x_{i}(t)), i\in \mathbb{Z}_{1}^{M}$. According to Assumptions~\ref{assump:L}, \ref{assump:dp-noise}, and the conditions that $\sum_{k=0}^{\infty}\chi^k=\infty$, $\sum_{k=0}^{\infty}(\chi^k)^2<\infty$, $\sum_{k=0}^{\infty}\frac{(\gamma^k)^2}{\chi^k}<\infty$, and $\|g_{i}(\tilde{{\bm{u}}}_{i}) \|\le C_{g}$, we can follow the same line of reasoning as that of Theorem 1 in \cite{WangTAC2024} to obtain the results.
\end{proof}
\vspace{-5pt}
The following lemma is required for convergence analysis:
\vspace{-15pt}
\begin{lemma}[Lemma 11, \cite{Polyak1987}]\label{Lemma-polyak1}
Let $\{\psi^k\}$, $\{\phi^k\}$, $\{a^k\}$, and $\{\varpi^k\}$ be random non-negative scalar sequences such that
\[
\begin{aligned}
	\mathbb{E}\left[\psi^{k+1}|\mathcal{F}^k \right]\le(1+{a^k}) \psi^k -\phi^k+{ \varpi^k},\quad \forall k\geq 0,
\end{aligned}
\]
where $\mathcal{F}^k=\{\psi^\ell,\phi^\ell,{a^\ell},{\varpi^\ell};\, 0\le \ell\le k\}$. If $\sum_{k=0}^\infty {a^k}<\infty$ and $\sum_{k=0}^\infty {\varpi^k}<\infty$, then $\sum_{k=0}^{\infty}\phi^k<\infty$ and $\{ \psi^k \}$ converges to a finite variable almost surely.
\end{lemma}
\vspace{-5pt}
\begin{theorem}\label{Th:convergence}
Suppose Assumptions~\ref{assump:system}, \ref{assump:L}, and \ref{assump:dp-noise} hold. If
the non-negative sequences $\{\chi^k\}$ and $\{\gamma^k\}$ satisfy
$\sum_{k=0}^{\infty}\chi^k=\infty$, $\sum_{k=0}^{\infty}(\chi^k)^2<\infty$, $\sum_{k=0}^{\infty}\gamma^k=\infty$, and $\sum_{k=0}^{\infty}\frac{(\gamma^k)^2}{\chi^k}<\infty$, then Algorithm~\ref{algo:dual_privacy} guarantees that $\lim_{k\rightarrow \infty} \mathcal{L}(\{\tilde{{\bm{u}}}_{i}^{*}\}, \bar{\lambda}^{k}) = \mathcal{L}(\{\tilde{{\bm{u}}}_{i}^{*}\}, \lambda^{*})$ and $\lim_{k\rightarrow \infty} \mathcal{L}(\{\tilde{{\bm{u}}}_{i}^{k}\}, \lambda^{*}) = \mathcal{L}(\{\tilde{{\bm{u}}}_{i}^{*}\}, \lambda^{*})$
hold almost surely.
\end{theorem}
\vspace{-5pt}
\begin{proof}
Based on Lemma 1 in~\cite{NedicTAC2010} and the update law of $\lambda_{i}^{k}$ in \eqref{equ:DP_DDGA_2}, it can be obtained that for any {\small$\lambda\in \mathbb{R}_{+}^{Np}$}, {\small$\sum_{i=1}^{M} \| \lambda_{i}^{k\!+\!1}\!-\!\lambda \|^{2} = \sum_{i=1}^{M} \| \Pi_{\mathbb{R}^{Np}_{+}} [\lambda_{i}^{k} \!+\!\chi^{k} \sum_{j\in \mathcal{N}_{i}} L_{ij}(\hat{\lambda}_{j}^{k}\!-\!\lambda_{i}^{k}) \!+\! \gamma^{k}g_{i}(\tilde{{\bm{u}}}_{i}^{k\!+\!1})]\!-\!\lambda \|^{2} \le \sum_{i=1}^{M} \| \lambda_{i}^{k}\!+\!\chi^{k} \sum_{j\in \mathcal{N}_{i}} L_{ij}(\hat{\lambda}_{j}^{k}\!-\!\lambda_{i}^{k}) \!+\! \gamma^{k}g_{i}(\tilde{{\bm{u}}}_{i}^{k\!+\!1}) \!-\! \lambda \|^{2} \le \sum_{i=1}^{M} \| \lambda_{i}^{k} \!+\! \chi^{k} \sum_{j\in \mathcal{N}_{i}} L_{ij}(\lambda_{j}^{k}\!+\!\zeta_{j}^{k}\!-\!\lambda_{i}^{k}) \!+\! \gamma^{k}g_{i}(\tilde{{\bm{u}}}_{i}^{k\!+\!1})  \!-\!\lambda \|^{2} \le \sum_{i=1}^{M}\| \sum_{j\in \mathcal{N}_{i}\cup \{i\}}w_{ij}^{k}\lambda_{j}^{k} \!-\! \lambda \!+\! \chi^{k}\xi_{i}^{k} \!+\! \gamma^{k}g_{i}(\tilde{{\bm{u}}}_{i}^{k\!+\!1}) \|^{2}$}, where $w_{ij}^{k}$ and $\xi_{i}^{k}$ are defined as
\begin{equation} \label{equ:wij}
	{\small
		\begin{aligned}
			w_{ii}^{k}:=1+\chi^{k}L_{ii}, \quad w_{ij}^{k}:=\chi^{k}L_{ij}, \quad \xi_{i}^{k}:=\sum_{j\in \mathcal{N}_{i}} L_{ij}\zeta_{j}^{k}.
		\end{aligned}
	}
\end{equation}
Using $w_{ij}^{k}\lambda_{j}^{k} \!-\! \lambda=(\bar{\lambda}^{k}\!-\!\lambda)+ (w_{ij}^{k}\lambda_{j}^{k} \!-\! \bar{\lambda}^{k})$, it can be further concluded that
\begin{equation} \label{equ:lambda}
	{\small
		\begin{aligned}
			&\sum_{i=1}^{M}\! \| \lambda_{i}^{k\!+\!1}\!-\!\lambda \|^{2} 
			\!\!\le\! \sum_{i=1}^{M}\! \left( \| \sum_{j\in \mathcal{N}_{i}\cup \{i\}} \!\!\! w_{ij}^{k}\lambda_{j}^{k} \!-\! \lambda \|^{2} \!\!\!+\! \| \chi^{k}\xi_{i}^{k} \!+\! \gamma^{k}g_{i}(\tilde{{\bm{u}}}_{i}^{k\!+\!1}) \|^{2} \right.
			\\
			& \qquad \left.+ 2 \! \left( \sum_{j\in \mathcal{N}_{i}\cup \{i\}}\!\!\! w_{ij}^{k}\lambda_{j}^{k} \!-\! \lambda  \right)^{\top}\! \left( \chi^{k}\xi_{i}^{k} \right) \right.
			\left.\!\!+\! 2\left( \bar{\lambda}^{k} \!-\! \lambda  \right)^{\top}\!\left(\gamma^{k}g_{i}(\tilde{{\bm{u}}}_{i}^{k\!+\!1}) \right) \right.
			\\
			& \qquad \left. + 2 \! \left( \sum_{j\in \mathcal{N}_{i}\cup \{i\}}\!\!\! w_{ij}^{k}\lambda_{j}^{k} \!-\! \bar{\lambda}^{k}  \right)^{\top}\! \left(\gamma^{k}g_{i}(\tilde{{\bm{u}}}_{i}^{k\!+\!1}) \right) \right).
		\end{aligned}
	}
\end{equation}

According to Assumptions~\ref{assump:L}, \ref{assump:dp-noise} and \eqref{equ:wij}, one can verify that 
\begin{align}
	&w_{ij}^{k}=w_{ji}^{k}, \quad \sum_{i=1}^{M}w_{ij}^{k}=\sum_{j=1}^{M}w_{ij}^{k}=\sum_{j\in \mathcal{N}_{i}\cup \{i\}}\!\!\!w_{ij}^{k}=1, \label{equ:wij_sum}
	\\
	&\mathbb{E}\left[\xi_i^k\mid \lambda_i^k\right]=0, \quad \mathbb{E}\left[\|\xi_i^k\|^2\mid \lambda_{i}^{k}\right]= \sum_{j\in \mathcal{N}_{i}}(L_{ij} \sigma_{j}^k)^2. \label{equ:xi}
\end{align}
By~\eqref{equ:wij_sum} and by the convexity of $\|\cdot\|^{2}$, we have that
\begin{equation} \label{equ:lambda2}
	{\small
		\begin{aligned}
			&\sum_{i=1}^{M} \| \sum_{j\in \mathcal{N}_{i}\cup \{i\}} \!\!\!w_{ij}^{k}\lambda_{j}^{k} \!-\! \lambda \|^{2} \!=\! \sum_{i=1}^{M} \| \sum_{j\in \mathcal{N}_{i}\cup \{i\}}\!\!\!w_{ij}^{k} \left(\lambda_{j}^{k} \!-\! \lambda\right) \|^{2}
			\\
			&\le \sum_{i=1}^{M} \sum_{j\in \mathcal{N}_{i}\cup \{i\}} \!\!\! w_{ij}^{k} \|  \left(\lambda_{j}^{k} - \lambda\right) \|^{2} \le \sum_{i=1}^{M} \| \lambda_{i}^{k} - \lambda \|^{2}.
		\end{aligned}
	}
\end{equation}
It can be obtained from~\eqref{equ:DP_DDGA_1} that for any {\small$\tilde{{\bm{u}}}_{i}\!\in \! \tilde{\mathcal{U}_{i}}(x_{i}(t))$}, {\small$J_{i}(x_{i}(t), \tilde{{\bm{u}}}^{k+1}_{i}) \!+\! (\lambda_{i}^{k})^{\top}g_{i}(\tilde{{\bm{u}}}^{k+1}_{i}) \le J_{i}(x_{i}(t), \tilde{{\bm{u}}}_{i}) + (\lambda_{i}^{k})^{\top}g_{i}(\tilde{{\bm{u}}}_{i})$}. Thus, we can further derive that
\begin{equation} \label{equ:lambda3}
	{\small
		\begin{aligned}
			&\sum_{i=1}^{M} \left( \bar{\lambda}^{k} - \lambda  \right)^{\top}\left(\gamma^{k}g_{i}(\tilde{{\bm{u}}}_{i}^{k+1}) \right)
			\\
			=&\gamma^{k}\sum_{i=1}^{M} \left( \left( \bar{\lambda}^{k} - \lambda_{i}^{k}  \right)^{\top}g_{i}(\tilde{{\bm{u}}}_{i}^{k+1}) + \left(\lambda_{i}^{k}-\lambda\right)^{\top}g_{i}(\tilde{{\bm{u}}}_{i}^{k+1}) \right.
			\\
			& \qquad  \left. + J_{i}(x_{i}(t), \tilde{{\bm{u}}}_{i}^{k+1}) - J_{i}(x_{i}(t), \tilde{{\bm{u}}}_{i}^{k+1}) \right)
			\\
			\le&\gamma^{k}\sum_{i=1}^{M} \left( \left( \bar{\lambda}^{k} - \lambda_{i}^{k}  \right)^{\top}g_{i}(\tilde{{\bm{u}}}_{i}^{k+1}) + \left(\lambda_{i}^{k}-\bar{\lambda}^{k}\right)^{\top}g_{i}(\tilde{{\bm{u}}}_{i}) \right)
			\\
			& +\gamma^{k}\left( \mathcal{L}(\{\tilde{{\bm{u}}}_{i}\}, \bar{\lambda}^{k})- \mathcal{L}(\{\tilde{{\bm{u}}}_{i}^{k+1}\}, \lambda)  \right).
		\end{aligned}
	}
\end{equation}
Using \eqref{equ:xi}-\eqref{equ:lambda3} and the fact that $\|g_{i}(\tilde{{\bm{u}}}_{i})\|\le C_{g}$, $\forall \tilde{{\bm{u}}}_{i}\in \tilde{\mathcal{U}_{i}}(x_{i}(t))$, we can take the conditional expectation with respect to $\mathcal{F}^{k}=\{ \lambda_{\ell}^{k}, \tilde{\bm{u}}_{\ell}^{k+1}; 0\le \ell \le k\}$ in \eqref{equ:lambda} to obtain
\begin{equation} \label{equ:E}
	{\small
		\begin{aligned}
			&\sum_{i=1}^{M} \mathbb{E}\left[ \| \lambda_{i}^{k+1}\!-\!\lambda \|^{2} | \mathcal{F}^{k} \right]
			\\
			\le &\sum_{i=1}^{M} \| \lambda_{i}^{k} \!-\! \lambda \|^{2} \!+\! d^{k} \!+\! 2\gamma^{k}\left( \mathcal{L}(\{\tilde{{\bm{u}}}_{i}\}, \bar{\lambda}^{k})\!-\! \mathcal{L}(\{\tilde{{\bm{u}}}_{i}^{k\!+\!1}\}, \lambda)  \right),
		\end{aligned}
	}
\end{equation}
where {\small$d^{k}\!=\!(\chi^{k})^{2}\sum_{i=1}^{M}\sum_{j\in \mathcal{N}_{i}}(L_{ij} \sigma_{j}^k)^2 + M(\gamma^{k})^{2}C_{g}^{2}+6C_{g}\gamma^{k}\sum_{i=1}^{M} \|\lambda_i^k-\bar{\lambda}^{k}\|$}.
Based on Assumption~\ref{assump:dp-noise}, Theorem~\ref{Th:consensus_tracking}, and the conditions for $\chi^{k}$ and $\gamma^{k}$, it can be concluded that $d^{k}$ is summable, i.e., $\sum_{k=0}^{\infty}d^{k}<\infty$.

Plugging the optimal primal-dual pair $(\{\tilde{{\bm{u}}}_{i}^{*}\}, \lambda^{*})$ into~\eqref{equ:E} and utilizing the Saddle-Point Theorem~\eqref{equ:saddle}, we can arrive at
\begin{equation} \label{equ:E2}
	{\small
		\begin{aligned}
			&\sum_{i=1}^{M} \mathbb{E}\left[ \| \lambda_{i}^{k+1}-\lambda^{*} \|^{2} | \mathcal{F}^{k} \right]
			\\
			\le &\sum_{i=1}^{M} \| \lambda_{i}^{k} \!-\! \lambda^{*} \|^{2} \!\!+\! d^{k} \!+\!2\gamma^{k} \left(\! \mathcal{L}(\{\tilde{{\bm{u}}}_{i}^{*}\}, \bar{\lambda}^{k})\!-\! \mathcal{L}(\{\tilde{{\bm{u}}}_{i}^{*}\}, \lambda^{*}) \!\right),
		\end{aligned}
	}
\end{equation}
\begin{equation} \label{equ:E3}
	{\small
		\begin{aligned}
			&\sum_{i=1}^{M} \mathbb{E}\left[ \| \lambda_{i}^{k+1}-\lambda^{*} \|^{2} | \mathcal{F}^{k} \right]
			\\
			\le &\sum_{i=1}^{M}\! \| \lambda_{i}^{k} \!-\! \lambda^{*} \|^{2} \!\!+\! d^{k} \!+\!2\gamma^{k}\!\left(\! \mathcal{L}(\{\tilde{{\bm{u}}}_{i}^{*}\}, \lambda^{*})\!-\! \mathcal{L}(\{\tilde{{\bm{u}}}_{i}^{k\!+\!1}\}, \lambda^{*}) \! \right).
		\end{aligned}
	}
\end{equation}
According to Lemma~\ref{Lemma-polyak1} and {\small$\sum_{k=0}^{\infty}d^{k}<\infty$}, it can be concluded that {\small$\gamma^{k} \left( \mathcal{L}(\{\tilde{{\bm{u}}}_{i}^{*}\}, \bar{\lambda}^{k})- \mathcal{L}(\{\tilde{{\bm{u}}}_{i}^{*}\}, \lambda^{*})  \right)$} in~\eqref{equ:E2} and {\small$\gamma^{k} \left( \mathcal{L}(\{\tilde{{\bm{u}}}_{i}^{*}\}, \lambda^{*})- \mathcal{L}(\{\tilde{{\bm{u}}}_{i}^{k+1}\}, \lambda^{*})  \right)$} in~\eqref{equ:E3}  satisfy the conditions for $\phi^{k}$ in Lemma~\ref{Lemma-polyak1}, i.e., the following relationships hold almost surely:
\begin{equation}
	{\small
		\begin{aligned}
			&\sum_{k=1}^{\infty}\gamma^{k} \left( \mathcal{L}(\{\tilde{{\bm{u}}}_{i}^{*}\}, \bar{\lambda}^{k})- \mathcal{L}(\{\tilde{{\bm{u}}}_{i}^{*}\}, \lambda^{*})  \right) < \infty,
			\\
			&\sum_{k=1}^{\infty}\gamma^{k} \left( \mathcal{L}(\{\tilde{{\bm{u}}}_{i}^{*}\}, \lambda^{*})- \mathcal{L}(\{\tilde{{\bm{u}}}_{i}^{k+1}\}, \lambda^{*})  \right) < \infty.
		\end{aligned}
	}
\end{equation}
Since $\gamma^{k}$ is non-summable, we have that {\small$\mathcal{L}(\{\tilde{{\bm{u}}}_{i}^{*}\}, \bar{\lambda}^{k})- \mathcal{L}(\{\tilde{{\bm{u}}}_{i}^{*}\}, \lambda^{*})$} and {\small$\mathcal{L}(\{\tilde{{\bm{u}}}_{i}^{*}\}, \lambda^{*})- \mathcal{L}(\{\tilde{{\bm{u}}}_{i}^{k+1}\}, \lambda^{*})$} converge to zero almost surely.
\end{proof}
\vspace{-10pt}
\begin{remark} \label{remark:chi_gamma}
The weakening factor sequence $\{\chi^{k}\}$ is employed to mitigate the influence of persistent DP noise, and the step-size sequence $\{\gamma^{k}\}$ should be appropriately selected to work with $\{\chi^{k}\}$ for guaranteed convergence. 
The conditions for the sequences $\{\chi^{k}\}$ and $\{\gamma^{k}\}$ in Theorems~\ref{Th:consensus_tracking} and \ref{Th:convergence} can be satisfied, e.g., by selecting $\chi^{k} = \frac{c_1}{1+c_{2}k^{c_3}}$ and $\gamma^{k}=\frac{c_4}{1+c_{5}k}$ with any $c_{1}>0$, $c_{2}>0$, $0.5<c_{3}<1$, $c_{4}>0$, and $c_{5}>0$. Note that the design of $\chi^{k}$ in this example is identical to the one in Remark~\ref{remark:DP-noise}. Therefore, the sequences $\{\chi^{k}\}$, $\{\gamma^{k}\}$, and $\{\nu^{k}\}$ can be meticulously tailored to meet all conditions required by Assumption~\ref{assump:dp-noise} and Theorems~\ref{Th:consensus_tracking}, \ref{Th:convergence}.
\end{remark}

\vspace{-15pt}
\subsection{Privacy Analysis}
Using the adjacency concept defined in Definition~\ref{de:adjacency}, we establish two adjacent distributed optimization problems, denoted as $\mathcal{P}$ and $\mathcal{P}'$. There is only one signal that differs between these two problems, and without loss of generality we denote it as $g_i$ in $\mathcal{P}$ and $g'_i$ in $\mathcal{P}'$. 
Per the third condition of Definition~\ref{de:adjacency}, the signals $g_i$ and $g'_i$ should exhibit similar behavior near the optimal solution. Specifically, $g_i$ and $g'_i$ should converge toward each other if the algorithm guarantees convergence to the optimal solution. 
Leveraging the proven convergence from Theorem~\ref{Th:convergence}, we formalize this condition by requiring the existence of a constant
$C>0$ such that:
\begin{equation}\label{eq:r_convergence_rate}
\|g_i(\tilde{\bm{u}}^{k+1}_{i})-{g^{'}_i(\tilde{\bm{u}}'^{k+1}_{i})} \|_1\leq C\chi^k
\end{equation}
holds for all $k\geq 0$.

For Algorithm~\ref{algo:dual_privacy}, an execution is represented as {\small$\mathcal{A}=\{\vartheta^0,\vartheta^1,\ldots\}$} with {\small$\vartheta^k=\lambda^k = \begin{bmatrix}
	(\lambda_{1}^{k})^{\top}, \cdots, (\lambda_{M}^{k})^{\top}
\end{bmatrix}^{\top}$}. An observation sequence is denoted as {\small$\mathcal{O}=\{o^0,o^1,\ldots\}$} with {\small$o^k=\hat{\lambda}^k= \begin{bmatrix}
	(\hat{\lambda}_{1}^{k})^{\top}, \cdots, (\hat{\lambda}_{M}^{k})^{\top}
\end{bmatrix}^{\top}$} (note that {\small$\hat{\lambda}_i^k=\lambda_i^k+\zeta_i^{k}$}, as detailed in Algorithm~\ref{algo:dual_privacy}). Similar to the sensitivity metric for constraint-free distributed optimization in \cite{Huang2015}, we formulate the sensitivity of Algorithm~\ref{algo:dual_privacy} as follows:
\vspace{-5pt}
\begin{definition}\label{de:sensitivity}
At each iteration $k$, for any two adjacent distributed optimization problems $\mathcal{P}$ and $\mathcal{P}^{'}$ and any initial state $\vartheta^0$, the sensitivity of Algorithm~\ref{algo:dual_privacy} is given by
\begin{equation}
	{\small
		\begin{aligned}
			\Delta^k:= \sup\limits_{\mathcal{O}\in\mathbb{O}}\left\{\sup\limits_{\vartheta\in\mathcal{R}_{\mathcal{P},\vartheta^0}^{-1}(\mathcal{O}),\:\vartheta^{'}\in\mathcal{R}_{\mathcal{P}^{'},\vartheta^0}^{-1}(\mathcal{O})}\hspace{-0.3cm}\|\vartheta^{k}-\vartheta'^{k}\|_1\right\},
		\end{aligned}
	}
\end{equation}
where $\mathbb{O}$ denotes the set of all possible observation sequences.
\end{definition}
\vspace{-5pt}
Given Definition~\ref{de:sensitivity}, we have the following lemma:
\vspace{-5pt}
\begin{lemma}\label{Le:Laplacian}
In Algorithm~\ref{algo:dual_privacy}, at each iteration $k$, if each subsystem's DP noise vector $\zeta_i^k\in\mathbb{R}^{Np}$ comprises $Np$ independent Laplace noises with parameter $\nu^k$, satisfying $\sum_{k=1}^{T_0}\frac{\Delta^k}{\nu^k}\leq \bar\epsilon$ for some $\bar\epsilon>0$, then Algorithm~\ref{algo:dual_privacy} achieves $\epsilon$-differential privacy with the cumulative privacy level for iterations $0\leq k\leq T_0$ less than $\bar\epsilon$.
\end{lemma}
\vspace{-5pt}
\begin{proof}
The proof of this lemma follows the same reasoning as that of Lemma~2 in~\cite{Huang2015}.
\end{proof}
\vspace{-5pt}
We also introduce the following lemma for privacy analysis:
\vspace{-5pt}
\begin{lemma}(Lemma 4, \cite{WangTAC2023}) \label{le:chung}
Let $\{\psi^k\}$ be a non-negative sequence, and $\{a^k\}$ and $\{\varpi^k\}$ be positive sequences satisfying $\sum_{k=0}^{\infty}a^k=\infty$, $\lim_{k\rightarrow \infty} a^k =0$, and $\frac{\varpi^k}{a^k}$ converges to zero with a polynomial rate. If there exists a $\bar{K}\geq 0$ such that $ \psi^{k+1} \le(1-a^k) \psi^k +\varpi^k$ holds for all $k\geq \bar{K}$,
then it follows that $\psi^k\leq \bar{C} \frac{\varpi^k}{a^k}$ for all $k$, with $\bar{C}$ being some constant.
\end{lemma}
\vspace{-5pt}
\begin{theorem}\label{th:DP_Algorithm1}
Suppose the conditions of Theorem \ref{Th:consensus_tracking} hold. If every element of $\zeta_i^k$ is independently sampled from Laplace distribution ${\rm Lap}(\nu^k)$, where $(\sigma_i^k)^2=2(\nu^k)^2$ satisfies Assumption \ref{assump:dp-noise}, then the following results hold: $\left. 1 \right)$ For any finite number of iterations $T$, Algorithm~\ref{algo:dual_privacy} ensures  $\epsilon$-differential privacy, and the cumulative privacy budget is bounded by $\epsilon\leq \sum_{k=1}^{T}\frac{C\nu^k}{\nu^k}$, where $\varsigma^k:= \sum_{s=1}^{k-1} \Pi_{q=s}^{k-1}(1\!-\!\chi^q\bar{L}) \gamma^{s-1}\chi^{s-1} \!+\!\gamma^{k-1}\chi^{k-1}$, $\bar{L}:=\min_{i\in \mathbb{Z}_{1}^{M}}|L_{ii}|$, and $C$ is from \eqref{eq:r_convergence_rate}; $\left. 2 \right)$ If $\sum_{k=0}^{\infty}\frac{\gamma^k}{\nu^k}<\infty$ holds, the cumulative privacy budget remains finite as $T\rightarrow\infty$.
\end{theorem}
\vspace{-5pt}
\begin{proof}
To establish the privacy guarantees, we begin by analyzing the sensitivity of Algorithm~\ref{algo:dual_privacy}. Given any initial state $\lambda^0$, any fixed observation $\mathcal{O}$, and two adjacent distributed optimization problems $\mathcal{P}$ and $\mathcal{P}^{'}$, the sensitivity depends on $\|\lambda^{k}-\lambda'^{k}\|_1$ as per Definition \ref{de:sensitivity}. Note that $\mathcal{P}$ and $\mathcal{P}^{'}$ differ solely in one signal, and without loss of generality, we denote this distinct signal as the $i$th one, i.e., $g_i$ in $\mathcal{P}$ and $g^{'}_i$ in $\mathcal{P}^{'}$.
Since the initial conditions and observations of $\mathcal{P}$ and $\mathcal{P}^{'}$ are the same for $j\neq i$, it follows that $\lambda_j^k=\lambda'^{k}_j$ for all $k$ and $j\neq i$. Consequently, $\|\lambda^{k}-\lambda'^{k}\|_1$ is always equal to $\|\lambda_i^{k}-\lambda'^{k}_i\|_1$.

Based on \eqref{equ:DP_DDGA_2} in Algorithm~\ref{algo:dual_privacy}, $L_{ii}:=-\sum_{j\in\mathbb{N}_i}L_{ij}$, and the fact that the observations $\lambda_j^k+\zeta_j^k$ and ${\lambda'_j}^k+{\zeta'_j}^k$ are identical, we can derive that
\begin{equation} \label{equ:delta_lambda}
	{\small
		\begin{aligned}
			\|\lambda_i^{k+1}-{\lambda'_i}^{k+1}\|_1
			\leq  &(1-|L_{ii}|\chi^k)\|\lambda_i^k-{\lambda'_i}^k\|_1 
			\\
			&\quad 
			+\gamma^k \|g_i(\tilde{\bm{u}}_{i}^{k+1})-{g^{'}_i(\tilde{\bm{u}}'^{k+1}_{i})}\|_1.
		\end{aligned}
	}
\end{equation}
From \eqref{eq:r_convergence_rate} and \eqref{equ:delta_lambda}, it follows that 
\begin{equation}\label{eq:delta1}
{\small
	\begin{aligned}
		\Delta^{k+1}
		\leq (1-|L_{ii}|\chi^k)\Delta^{k}+C \gamma^k \chi^k.
	\end{aligned}
}
\end{equation}
Using Lemma~\ref{Le:Laplacian} and~\eqref{eq:delta1}, the first statement is established.

Lemma~\ref{le:chung} is applied to prove the second statement of Theorem~\ref{th:DP_Algorithm1}. Specifically, based on \eqref{eq:delta1} and the properties of $\chi^k$ and $\gamma^k$, Lemma~\ref{le:chung} implies that there exists some constant $\bar{C}$ such that the sensitivity $\Delta^{k}$ satisfies $\Delta^{k}\le \bar{C}\gamma^k$. It can be further obtained from Lemma \ref{Le:Laplacian} that $\epsilon\leq \sum_{k=1}^{T}\frac{\bar{C}\gamma^k}{\nu^k}$. Thus, if $\sum_{k=0}^{\infty}\frac{\gamma^k}{\nu^k}<\infty$ holds (i.e., the sequence $\{\frac{\gamma^k}{\nu^k}\}$ is summable), then $\epsilon$ will be finite even when $T\rightarrow\infty$.
\end{proof}
\vspace{-5pt}
For $\epsilon$-differential privacy (see Definition~\ref{de:differential_privacy}), a smaller $\epsilon$ indicates a better extent of privacy preservation. According to Theorem~\ref{th:DP_Algorithm1}, for given $C$ and $\varsigma^k$, a higher noise level $\nu^{k}$ results in a smaller $\epsilon$, thereby enhancing privacy protection.
\vspace{-5pt}
\begin{remark} \label{remark:overhead}
The DP noise injection mechanism in Algorithm~\ref{algo:dual_privacy} is computationally efficient and easy to implement. To mitigate the impact of DP noise, careful design of the weakening factor and step-size sequences is essential. When the DP noise intensity is high, the algorithm may require more iterations to converge compared to non-privacy-preserving methods. This trade-off is necessary to achieve both $\epsilon$-differential privacy and provable convergence to the optimal solution.
\end{remark}
\vspace{-10pt}
\section{Implementation of Privacy-Preserving DMPC} \label{sec:MPC}
In this section, the overall implementation strategy of DMPC is described.
\vspace{-5pt}
\subsection{Algorithm Implementation}
Algorithm~\ref{algo:dual_privacy} will terminate after $\bar{k}$ iterations. Note that Algorithm~\ref{algo:dual_privacy} converges almost surely in a probabilistic sense, and thus the global constraints~\eqref{equ:globalConstraint} may not necessarily be satisfied within a finite number of iterations. Based on~\eqref{equ:fb}, one can verify that the global constraints are satisfied if the following condition holds:
\begin{equation} \label{equ:g_constraint}
{\small
	\begin{aligned}
		\sum_{i=1}^{M}g_{i}(\tilde{{\bm{u}}}_{i}^{\bar{k}}) = \sum_{i=1}^{M}f_{i}(x_{i}(t), \tilde{{\bm{u}}}_{i}^{\bar{k}}) - b(\varepsilon) \le \varepsilon M {\bf 1}_{Np}.
	\end{aligned}
}
\end{equation}
To verify whether the global constraints are satisfied after the termination of Algorithm~\ref{algo:dual_privacy}, we employ a privacy-preserving static average consensus method developed in~\cite{WangTAC2019}.

Specifically, after Algorithm~\ref{algo:dual_privacy} terminates, each subsystem initializes $z_{i}^{0} = g_{i}(\tilde{{\bm{u}}}_{i}^{\bar{k}})= f_{i}(x_{i}(t), \tilde{{\bm{u}}}_{i}^{\bar{k}}) - \frac{b(\varepsilon)}{M}$. Then, $z_{i}^{0}$ is decomposed into two substates $z_{i, \alpha}^{0}$ and $z_{i, \beta}^{0}$, where $z_{i, \alpha}^{0}$ and $z_{i, \beta}^{0}$ are randomly chosen from the set of all real numbers with the constraint $z_{i, \alpha}^{0}+z_{i, \beta}^{0}=2z_{i}^{0}$. The static average consensus method updates $z_{i,\alpha}^{\ell}$ and $z_{i,\beta}^{\ell}$ as follows:
\begin{equation} \label{equ:static_consensus}
{\small
	\begin{aligned}
		z_{i,\alpha}^{\ell+1} &= z_{i,\alpha}^{\ell} + \iota \sum_{j\in \mathcal{N}_{i}} a_{ij}^{\ell}(z_{j,\alpha}^{\ell}-z_{i,\alpha}^{\ell}) + \iota a_{i,\alpha \beta}^{\ell}(z_{i,\beta}^{\ell}-z_{i,\alpha}^{\ell}),
		\\
		z_{i,\beta}^{\ell+1} &= z_{i,\beta}^{\ell} + \iota a_{i,\alpha \beta}^{\ell}(z_{i,\alpha}^{\ell}-z_{i,\beta}^{\ell}),
	\end{aligned}
}
\end{equation}
where $\iota$, $a_{i,\alpha \beta}^{\ell}$, $a_{i,j}^{\ell} \in \mathbb{R}_{+}$. As proven in \cite{WangTAC2019}, by appropriately selecting the parameters $\iota$, $a_{i,\alpha \beta}^{\ell}$, and $a_{i,j}^{\ell}$, $z_{i,\alpha}^{\ell}$ and $z_{i,\beta}^{\ell}$ converge to the average consensus value $\frac{1}{M} \sum_{i=1}^{M} z_{i}^{0}$ (i.e., $\frac{1}{M}\sum_{i=1}^{M} g_{i}(\tilde{{\bm{u}}}_{i}^{\bar{k}})$). Therefore, each subsystem can utilize the converged value of $z_{i,\alpha}^{\ell}$ to check whether condition~\eqref{equ:g_constraint} is satisfied. While conventional static average consensus methods~\cite{Olfati-Saber2007, SundaramACC2007, HendrickxTAC2015} can compute $\frac{1}{M} \sum_{i=1}^{M} z_{i}^{0}$ in a distributed manner, they require subsystems to directly share $z_{i}^{0}$ with neighbors. Since 
$z_{i}^{0}=g_{i}(\tilde{{\bm{u}}}_{i}^{\bar{k}})$ contains sensitive information about 
$\tilde{{\bm{u}}}_{i}^{\bar{k}}$, these methods may potentially lead to privacy breaches. The privacy-preserving average consensus method in~\cite{WangTAC2019} addresses this issue using a state decomposition scheme to mask the true values of $z_{i}^{0}$. Specifically, as shown in~\eqref{equ:static_consensus}, the substate 
$z_{i,\alpha}^{\ell}$ governs internode interactions and is the only value visible to a subsystem's neighbors. Meanwhile, the other substate $z_{i,\beta}^{\ell}$ interacts solely with $z_{i,\alpha}^{\ell}$, remaining hidden from neighboring subsystems but still influencing 
$z_{i,\alpha}^{\ell}$'s evolution. 
This design ensures strong privacy protection. For more details, please refer to \cite{WangTAC2019}.

The overall DMPC strategy is presented in Algorithm~\ref{algo:DMPC_privacy}. After executing the static average consensus method, an update mechanism is designed for the control input sequence $\tilde{{\bm{u}}}_{i}(t)$. Based on the consensus results, if condition \eqref{equ:g_constraint} is met, the solution $\tilde{{\bm{u}}}_{i}^{\bar{k}}$ obtained at the current time instant is adopted as $\tilde{{\bm{u}}}_{i}(t)$. Otherwise, we implement a one-step time-shift on the previous control input sequence $\tilde{{\bm{u}}}_{i}(t-1)$ and append a terminal control action to update $\tilde{{\bm{u}}}_{i}(t)$, as shown in~\eqref{equ:u_update}. The sequence constructed via~\eqref{equ:u_update} is guaranteed to be feasible, which will be demonstrated in the first statement of Theorem~\ref{theorem:mpc}. Therefore, by combining the static average consensus method with the update mechanism for $\tilde{{\bm{u}}}_{i}(t)$, Algorithm~\ref{algo:DMPC_privacy} ensures that if the initial solution $\tilde{{\bm{u}}}_{i}^{\bar{k}}$ at $t=0$ is feasible, then feasible solutions will be maintained at all subsequent time steps—even in cases when Algorithm~\ref{algo:dual_privacy} fails to generate feasible solutions at some time instants or over multiple consecutive steps.
\begin{algorithm}[!t]
	\small
	\SetAlFnt{\small}
	\SetKwInOut{Parameter}{Parameter}
	\SetKwInOut{Output}{Output}
	\caption{Privacy-preserving DMPC Algorithm}
	\label{algo:DMPC_privacy}
	\SetAlgoLined
	At time instant $t$, every subsystem $i$ measures its state $x_{i}(t)$;
	
	Every subsystem $i$ computes $\tilde{{\bm{u}}}_{i}^{\bar{k}}$ by following Algorithm~\ref{algo:dual_privacy};
	
	Every subsystem $i$ runs the static average consensus algorithm~\eqref{equ:static_consensus} to obtain $\sum_{i=1}^{M} g_{i}(\tilde{{\bm{u}}}_{i}^{\bar{k}})$;
	
	\eIf{Condition~\eqref{equ:g_constraint} is satisfied}{Set current control input sequence $\tilde{{\bm{u}}}_{i}(t):=\{\tilde{u}_{i}(0|t), \tilde{u}_{i}(1|t), \cdots, \tilde{u}_{i}(N-1|t)\}$ as $\tilde{{\bm{u}}}_{i}(t)=\tilde{{\bm{u}}}_{i}^{\bar{k}}$;}{Use $\tilde{{\bm{u}}}_{i}(t-1)$ to update $\tilde{{\bm{u}}}_{i}(t)$, i.e.,
			\begin{equation}\label{equ:u_update}
				\begin{aligned}
					\tilde{{\bm{u}}}_{i}(t)=&\{\tilde{u}_{i}(1|t-1), \tilde{u}_{i}(2|t-1), \cdots,
					\\
					&\quad \tilde{u}_{i}(N-1|t-1), K_{i}\tilde{x}_{i}(N|t-1)\};
				\end{aligned}
	\end{equation}}
	
	Save $\tilde{{\bm{u}}}_{i}(t)$ in subsystem $i$; apply $\tilde{u}_{i}(0|t)$ to subsystem~$i$; 
	
	Wait for the next time instant; let $t=t+1$ and go to step 1.			
\end{algorithm}
\vspace{-5pt}
\subsection{Feasibility and Stability}
\begin{theorem} \label{theorem:mpc}
	Assume that $\tilde{{\bm{u}}}_{i}^{\bar{k}}$ generated from Algorithm~\ref{algo:dual_privacy} satisfies the global constraints at time instant $t=0$. Then, the following results hold: $\left. 1\right)$ If Algorithm~\ref{algo:DMPC_privacy} has a feasible solution at time instant $t$, then it has a feasible solution at $t+1$; $\left. 2 \right)$ $\sum_{i=1}^{M}J_{i}(x_{i}(t), \tilde{{\bm{u}}}_{i}(t)) - \sum_{i=1}^{M}J_{i}(x_{i}(t), \tilde{{\bm{u}}}_{i}^{*}) \le \eta$, where $\eta\in \mathbb{R}_{+}$ is a bounded constant; $\left. 3 \right)$ If $\{ x_{i}\in\mathbb{R}^{n_{i}}: \| x_{i} \|^{2}_{Q_{i}} \le \eta \}\subset \mathcal{X}^{f}_{i}$, then the state trajectory of each subsystem converges to the terminal set $\mathcal{X}^{f}_{i}$ in finite time.
\end{theorem}

\begin{proof}
	As shown in Algorithm~\ref{algo:DMPC_privacy}, the input sequence at time instant $t$ is denoted by $\tilde{{\bm{u}}}_{i}(t)=\{\tilde{u}_{i}(0|t), \tilde{u}_{i}(1|t), \cdots, \tilde{u}_{i}(N-1|t)\}$. Let $\tilde{{\bm{x}}}_{i}(t)=\{\tilde{x}_{i}(0|t), \tilde{x}_{i}(1|t), \cdots, \tilde{x}_{i}(N|t)\}$ be the corresponding predicted state sequence. Since $\tilde{{\bm{u}}}_{i}(t)$ is a feasible solution, it can be obtained from \eqref{equ:U}, \eqref{equ:fb}, and \eqref{equ:g_constraint} that $\tilde{{\bm{u}}}_{i}(t)\in \tilde{\mathcal{U}_{i}}(x_{i}(t))$ (i.e., $\tilde{x}_{i}(\ell|t)\!\in \! \mathcal{X}_{i}, \tilde{u}_{i}(\ell|t)\!\in \! \mathcal{U}_{i}, \tilde{x}_{i}(N|t)\!\in \! \mathcal{X}_{i}^{f}, \ell \!\in \! \mathbb{Z}_{0}^{N-1}$) and
	\begin{equation} \label{equ:global_constraint}
		{\small	
			\begin{aligned}
				\sum_{i=1}^{M} \varPsi_{x_{i}}\tilde{x}_{i}(\ell|t) +\varPsi_{u_{i}}\tilde{u}_{i}(\ell|t) \le (1-\varepsilon M\ell) {\bf 1}_{p}, \ell \in \mathbb{Z}_{0}^{N-1}.
			\end{aligned}
		}
	\end{equation}
	At time instant $t+1$, an input sequence $\hat{{\bm{u}}}_{i}(t+1)$ and its corresponding predicted state sequence $\hat{{\bm{x}}}_{i}(t+1)$ are defined~as
	\begin{equation} \label{equ:hat_u_x}
		{\small
			\begin{aligned}
				\hat{{\bm{u}}}_{i}(t\!+\!1) 
				\!= &\{\hat{u}_{i}(0|t+1), \hat{u}_{i}(1|t+1), \cdots, \hat{u}_{i}(N-1|t+1)\}
				\\ 
				\!= &\{\tilde{u}_{i}(1|t), \tilde{u}_{i}(2|t), \cdots, \tilde{u}_{i}(N-1|t), K_{i}\tilde{x}_{i}(N|t)\},
				\\
				\hat{{\bm{x}}}_{i}(t\!+\!1) 
				\!= &\{\hat{x}_{i}(0|t+1), \hat{x}_{i}(1|t+1), \cdots, \hat{x}_{i}(N|t+1)\}
				\\ 
				\!= &\{\tilde{x}_{i}(1|t), \tilde{x}_{i}(2|t), \cdots, \tilde{x}_{i}(N|t), (A_{i}\!+\!B_{i}K_{i})\tilde{x}_{i}(N|t)\}.
			\end{aligned}
		}
	\end{equation} 
	Based on \eqref{equ:terminal}, \eqref{equ:global_constraint}, and \eqref{equ:hat_u_x}, it can be concluded that {\small$\hat{{\bm{u}}}_{i}(t+1)\in \tilde{\mathcal{U}_{i}}(x_{i}(t+1))$}, {\small$\sum_{i=1}^{M} \varPsi_{x_{i}}\hat{x}_{i}(\ell|t\!+\!1) \!+\!\varPsi_{u_{i}}\hat{u}_{i}(\ell|t\!+\!1)\!=\! \sum_{i=1}^{M} \varPsi_{x_{i}}\tilde{x}_{i}(\ell\!+\!1|t) \!+\!\varPsi_{u_{i}}\tilde{u}_{i}(\ell\!+\!1|t)\le (1-\varepsilon M(\ell+1)) {\bf 1}_{p}, \ell \in \mathbb{Z}_{0}^{N-2}$}, and {\small$\sum_{i=1}^{M} \varPsi_{x_{i}}\hat{x}_{i}(N\!-\!1|t\!+\!1) \!+\!\varPsi_{u_{i}}\hat{u}_{i}(N\!-\!1|t\!+\!1)\!=\! \sum_{i=1}^{M} (\varPsi_{x_{i}} \!+\! \varPsi_{u_{i}}K_{i}) \tilde{x}_{i}(N|t)
		\le (1\!-\!\varepsilon MN) {\bf 1}_{p}$}.
	Therefore, $\hat{{\bm{u}}}_{i}(t+1)$ is a feasible solution at time instant $t+1$, which completes the proof for the first statement of Theorem~\ref{theorem:mpc}. From the above analysis, it follows that the control input sequence constructed in~\eqref{equ:u_update} is feasible. Thus, if $\tilde{{\bm{u}}}_{i}^{\bar{k}}$ computed by Algorithm~\ref{algo:dual_privacy} is feasible at $t=0$, then the update mechanism for $\tilde{{\bm{u}}}_{i}(t)$ in Algorithm~\ref{algo:DMPC_privacy} ensures the solution feasibility for the remaining duration.
	
	Due to the recursive feasibility, $x_{i}(t)$ remains within~the bounded set $\mathcal{X}_{i}$, and the solution $\tilde{{\bm{u}}}_{i}^{\bar{k}}$ generated by~Algorithm~\ref{algo:dual_privacy} is confined to the bounded set {\small $\tilde{\mathcal{U}_{i}}(x_{i}(t))$}. Thus, {\small $J_{i}(x_{i}(t), \tilde{{\bm{u}}}_{i}(t))$} is bounded, and there exists a positive bounded constant $\eta$ such that {\small $\sum_{i=1}^{M}J_{i}(x_{i}(t), \tilde{{\bm{u}}}_{i}(t)) \!-\! \sum_{i=1}^{M}J_{i}(x_{i}(t), \tilde{{\bm{u}}}_{i}^{*}) \le \eta$}.
	
	To prove the third statement, we first define a Lyapunov function $V(\{x_{i}(t)\}):= \sum_{i=1}^{M}J_{i}(x_{i}(t), \tilde{{\bm{u}}}_{i}^{*})$. According to the algebraic Riccati equation~\eqref{equ:ARE} and \eqref{equ:hat_u_x}, we have  
	\begin{equation} \label{equ:J}
		{\small
			\begin{aligned}
				J_{i}(x_{i}(t\!\!+\!\!1), \hat{{\bm{u}}}_{i}(t\!\!+\!\!1)) \!-\! J_{i}(x_{i}(t), \tilde{{\bm{u}}}_{i}(t)) 
				\!=\! \!-\| x_{i}(t) \|^{2}_{Q_{i}} \!\!-\! \| \tilde{u}_{i}(0|t)\|^{2}_{R_{i}}.
			\end{aligned}
		}
	\end{equation}
	$\hat{{\bm{u}}}_{i}(t\!+\!1)$ is a feasible solution  at $t+1$ but may not be optimal. Thus, we have {\small$V(\{x_{i}(t\!+\!1)\})\le \sum_{i=1}^{M} J_{i}(x_{i}(t\!+\!1), \hat{{\bm{u}}}_{i}(t\!+\!1))= \sum_{i=1}^{M} \left( J_{i}(x_{i}(t), \tilde{{\bm{u}}}_{i}(t)) \!-\! \| x_{i}(t) \|^{2}_{Q_{i}} \!-\! \| \tilde{u}_{i}(0|t)\|^{2}_{R_{i}} \right) \le \sum_{i=1}^{M} \left( J_{i}(x_{i}(t), \tilde{{\bm{u}}}_{i}(t)) \!-\! \| x_{i}(t) \|^{2}_{Q_{i}} \right)$}, where the equality condition is due to \eqref{equ:J}. By~Statement $\left. 2\right)$, we have {\small$V(\{x_{i}(t+1)\})\le V(\{x_{i}(t)\})+\eta - \sum_{i=1}^{M} \| x_{i}(t) \|^{2}_{Q_{i}}$}, which indicates that $x_{i}(t)$ converges to the bounded set $\{ \{x_{i}\}: \sum_{i=1}^{M} \| x_{i} \|^{2}_{Q_{i}} \le \eta \}$ in finite time. Considering the assumption that $\{ x_{i}\in\mathbb{R}^{n_{i}}: \| x_{i} \|^{2}_{Q_{i}} \!\le\! \eta \}\!\subset \! \mathcal{X}^{f}_{i}$, it can be concluded that $x_{i}(t)$ enters the terminal set $\mathcal{X}^{f}_{i}$ in finite time.
\end{proof}
\vspace{-10pt}
\section{Numerical Simulations} \label{sec:simulation}
In this section, simulation is conducted to demonstrate the performance of the developed method. A group of four linear time-invariant subsystems are considered. The interaction weight matrix $L$ is set as {\small$L_{12}=L_{21}=L_{14}=L_{41}=\frac{1}{4}$, $L_{23}=L_{32}=\frac{3}{8}$, $L_{34}=L_{43}=\frac{5}{16}$, $L_{13}=L_{31}=L_{24}=L_{42}=0$, $L_{11}=-\frac{1}{2}$, $L_{22}=-\frac{5}{8}$, $L_{33}=-\frac{11}{16}$}, and {\small$L_{44}=-\frac{9}{16}$}. 
The system matrices $A_{i}$ and $B_{i}$ are chosen as
\[
{\small
\begin{aligned}
	A_{i} \!=\! \begin{bmatrix}
		1 & 1 \\ 0 & 1
	\end{bmatrix}, B_{i} \!=\! \begin{bmatrix}
		1 \\ 1
	\end{bmatrix}, i\!=\!1, 3;
	A_{i} \!=\! \begin{bmatrix}
		2 & 1 \\ 0 & 1
	\end{bmatrix}, B_{i} \!=\! \begin{bmatrix}
		1 \\ 1
	\end{bmatrix}, i\!=\!2, 4. 
\end{aligned}
}
\]
For all subsystems, the local state and input constraint sets are selected as $\mathcal{X}_{i} \!=\! \{ x_{i}: -1\!\le \! x_{i}\! \le \! 1 \}$ and $\mathcal{U}_{i} \!=\! \{ u_{i}: -0.3\!\le \! u_{i}\!\le \! 0.3 \}$, respectively. The global constraint is $-0.65\!\le \! \sum_{i=1}^{4}u_{i}\!\le \! 0.65$. The weight matrices $Q_{i}$ and $R_i$ are set as $Q_{i}\!=\!I$ and $R_i\!=\!0.1$, respectively. The length of the prediction horizon is chosen as $N=5$. In Algorithm~\ref{algo:dual_privacy}, we inject Laplace noise with parameter $\nu^{k}\!=\!0.1\!+\!0.001k^{0.1}$. The weakening factor sequence and step-size sequence is set as $\chi^{k}=\frac{2}{1+0.01k^{0.9}}$ and $\gamma^{k}=\frac{5}{1+0.1k}$, respectively. In the simulation, Algorithm~\ref{algo:DMPC_privacy} is executed 20 times, and the mean and the variance of the state and input trajectories are computed. For comparison, we also run Algorithm~\ref{algo:DMPC}, which uses Algorithm~\ref{algo:dual} for distributed computation, and apply the same level of noise to the shared variables in Algorithm~\ref{algo:dual}.

The simulation results are illustrated in Fig.~\ref{fig_sim}. Fig.~\ref{fig_x1} depicts the state evolution of subsystem 1 (similar results for other subsystems are omitted). It can be seen that the variance of the system state trajectories under Algorithm~\ref{algo:DMPC} is much larger than those under Algorithm~\ref{algo:DMPC_privacy}. In addition, Fig.~\ref{fig_u} shows the evolution of the global constraint. It can be found that there exist constraint violations in Algorithm~\ref{algo:DMPC}. However, owing to the implementation scheme developed in Section~\ref{sec:MPC}, our approach can guarantee the satisfaction of the global constraint.

\begin{figure}[!t]
\vspace{-5pt}
\centering
\hspace{-5mm}
\subfigure[]{
	\includegraphics[width=1.8 in]{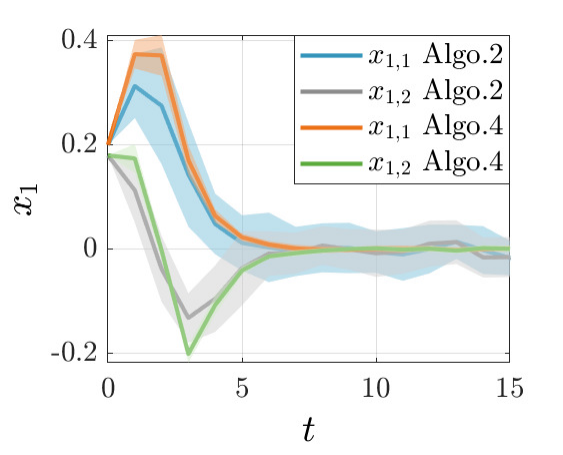} \label{fig_x1}
}
\hspace{-7mm}
\subfigure[]{
	\includegraphics[width=1.8 in]{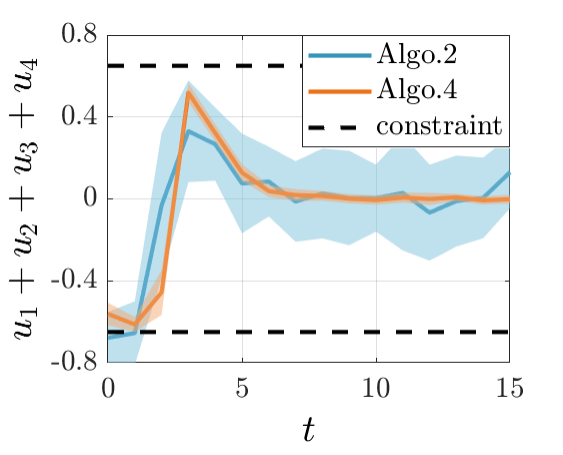} \label{fig_u}
}
\vspace{-5pt}
\caption{Evolution of (a) subsystem 1 and (b) global constraint.}
\label{fig_sim}
\end{figure}
\vspace{-5pt}
\section{Conclusion} \label{sec:conclusion}
This paper developed a differentially private DMPC strategy for linear discrete-time systems with coupled global constraints. We incorporated a DP noise injection mechanism into the distributed dual-gradient algorithm, enabling privacy preservation while maintaining accurate optimization convergence.  Furthermore, a practical implementation approach for DMPC was proposed, which guarantees the feasibility and stability of the closed-loop system. Simulation results validated the effectiveness of the developed privacy-preserving DMPC strategy. Future work will extend the differentially private framework to directed communication networks and systems with uncertainties (e.g., robust DMPCs), and will evaluate the proposed framework on practical applications.

\vspace{-10pt}
\bibliographystyle{IEEEtran}
\bibliography{IEEEabrv,reference}
\end{document}